\documentclass[letterpaper,twocolumn,10pt]{article}

\usepackage{hyperref}
\usepackage{usenix}

\usepackage{latexsym,amsmath,amsfonts,amssymb,stmaryrd,mathtools}
\usepackage{amsthm}
\usepackage{algorithm, algorithmicx}
\usepackage[noend]{algpseudocode}
\usepackage{graphicx}
\usepackage{wrapfig}
\usepackage{listings, fancyvrb}
\usepackage{multirow}
\usepackage{comment}
\usepackage{xspace}
\usepackage{pifont}
\usepackage{parcolumns}
\usepackage{graphicx}
\usepackage{enumerate}
\usepackage{enumitem}
\usepackage{wrapfig}
\usepackage{tabularx}
\usepackage{titlecaps}
\usepackage{subcaption}
\usepackage{tikz}
\usepackage{pgfplots}
\usepackage{booktabs}

\microtypecontext{spacing=nonfrench}

\newtheorem{theorem}{Theorem}[section]
\newtheorem{lemma}[theorem]{Lemma}

\makeatletter
\lst@Key{countblanklines}{true}[t]%
{\lstKV@SetIf{#1}\lst@ifcountblanklines}

\lst@AddToHook{OnEmptyLine}{%
	\lst@ifnumberblanklines\else%
	\lst@ifcountblanklines\else%
	\advance\c@lstnumber-\@ne\relax%
	\fi%
	\fi}
\makeatother

\usetikzlibrary{
    pgfplots.colorbrewer,
} 
\pgfplotsset{compat=newest}
\pgfplotsset{
    tick label style={font=\small},
    label style={font=\small},
    legend style={font=\footnotesize},
    every axis/.append style={line width=1pt},
    cycle list name=exotic,
    cycle list/.define={my marks}{
        every mark/.append style={solid,fill=\pgfkeysvalueof{/pgfplots/mark list fill}},mark=*\\
        every mark/.append style={solid,fill=\pgfkeysvalueof{/pgfplots/mark list fill}},mark=square*\\
        every mark/.append style={solid,fill=\pgfkeysvalueof{/pgfplots/mark list fill}},mark=triangle*\\
        every mark/.append style={solid,fill=\pgfkeysvalueof{/pgfplots/mark list fill}},mark=diamond*\\
    },
    mark list fill={.!75!white},
    cycle multiindex* list={
        exotic
            \nextlist
        my marks
            \nextlist
        linestyles*
            \nextlist
        very thick
            \nextlist
    },
}
\usepgfplotslibrary{groupplots}

\begin{document}

\date{}

\title{\Large \bf Velos: One-sided Paxos for RDMA applications}

\author{
{\rm Rachid Guerraoui}\\
EPFL
\and
{\rm Antoine Murat}\\
EPFL
\and
{\rm Athanasios Xygkis}\\
EPFL
}

\maketitle


\begin{abstract}
Modern data centers are becoming increasingly equipped with RDMA-capable NICs. These devices enable distributed systems to rely on algorithms designed for shared memory.
RDMA allows consensus to terminate within a few microsecond in failure-free scenarios, yet, RDMA-optimized algorithms still use expensive two-sided operations in case of failure.
In this work, we present a new leader-based algorithm for consensus based on Paxos that relies solely on one-sided RDMA verbs. Our algorithm decides in a single one-sided RDMA operation in the common case, and changes leader also in a single one-sided RDMA operation in case of failure.
We implement our algorithm in the form of an SMR system named Velos, and we evaluated our system against the state-of-the-art competitor Mu.
Compared to Mu, our solution adds a small overhead of $\approx 0.6 \mu s$ in failure-free executions and shines during failover periods during which it is 13 times faster in changing leader.
\end{abstract}
\section{Introduction}

RDMA is becoming increasingly popular in data centers. This networking technology is implemented in modern NICs and allows for Remote Direct Memory Access, where a server can access the memory of another without involving the CPU of the latter~\cite{rdma-manual}. As a result, RDMA paves the way for distributed algorithms in the shared-memory model that are no longer confined within a single physical server.

Communication over RDMA takes two forms. One form is one-sided communication which shares similar semantics to local memory accesses: a server directly performs READs and WRITEs to the memory of a remote server without involving the remote CPU. The other form is called two-sided communication that has similar semantics to message passing: the server SENDs a message to the remote server which involves its CPU to process the RECEIVEd message.

A significant body of work on consensus~\cite{cachin2011rachid} over RDMA has been conducted over the past decade~\cite{aguilera2020microsecond, poke2015dare, wang2017apus, derecho}. These solutions primarily focus on increasing the throughput and lowering the latency of common case executions, thus achieving consensus in the order of a few microseconds. Nevertheless, outside of their common case, these systems suffer from orders of magnitude higher failover times, ranging from 1ms\cite{aguilera2020microsecond} to tens or hundreds of ms~\cite{poke2015dare, wang2017apus}.

In this work we introduce a leader-based consensus algorithm based on the well-known Paxos~\cite{paxos} that is efficient both during the common-case as well as during failover. Our algorithm provides comparative performance to the state-of-the-art system Mu, i.e., it achieves consensus in under $1.9 \mu s$, making it only twice as slow as Mu. At the same time, our algorithm is 13 times faster than Mu in the event of a leader failure and manages to failover in under $65 \mu s$. To achieve this, our algorithm relies on one-sided RDMA WRITEs as well as Compare \& Swap (CAS), a capability that is present in RDMA NICs. 

The basic idea behind our algorithm is that the original Paxos algorithm contains RPCs that are simple enough to be replaced with CAS operations. The CAS operations are initiated by the leader and are executed on the memory of the participating consensus nodes (leader and followers). We first describe a version of our algorithm for single-shot consensus, where we provide proofs of correctness. Then, we continue by extending this version to a fully-fledged system, in which our algorithm takes a form similar to multishot Paxos\cite{lamport2001paxos}. In the case of a stable leader, it decides in a single CAS operation to a majority. In the event of a failure, the leader is changed in a single additional CAS operation to a majority.

The rest of this document is as follows: In section \ref{sec:related}, we present related work. In section \ref{sec:preliminaries}, we introduce the consensus problem and present the well-known Paxos algorithm. In section \ref{sec:cas-rpc}, we explain how to transform the RPC algorithm into a CAS-based one and we prove the correctness of this transformation. In section \ref{sec:practical-considerations}, we discuss further practical considerations, which are relevant in converting our single-shot consensus algorithm to a multi-shot one. In section \ref{sec:impl}, we discuss our implementation. In section \ref{sec:evaluation}, we evaluate the performance of our solution against Mu, the most recent state-of-the-art system that implements consensus.

\section{Related work} \label{sec:related}

\subsection*{Mu} Mu~\cite{aguilera2020microsecond} is an RDMA-powered State Machine Replication~\cite{boichat2003deconstructing} system that operates at the microsecond scale. 
Similarly to our system, it decides in a single amortized RDMA RTT.
This is achieved by relying on RDMA permissions~\cite{rdma-manual}. Mu ensures that at any time, at most one process can write to a majority and decide, which ensures consensus' safety. In case of leader failure, Mu requires permission changes that take $\approx 250$µs. Mu thus fails at guaranteeing microsecond decisions in case of failure.

\subsection*{APUS} APUS~\cite{wang2017apus} is a Paxos-based SMR system. It was tailored for RDMA. It doesn't use expensive permission changes but relies heavily on two-sided communication schemes. While it provides short failovers, its consensus engine involves heavy CPU usage at replicas and is significantly slower than Mu's.

\subsection*{Disk Paxos} Disk Paxos~\cite{gafni2003disk} observes that Paxos' acceptors can be replaced by moderately smart and failure-prone shared memories. Their work can be done by proposers as long as they are able to run atomic operations at the shared resource. This work is purely theoretical.
\section{Preliminaries} \label{sec:preliminaries}
In this section, we state the process and communication model that we assume for the rest of this work. Then we formally introduce the problem of consensus and we present the well-known Paxos algorithm. 

\subsection{Assumptions}\label{sec:communication}

We consider the message-and-memory (M\&M) model~\cite{aguilera2018passing}, which allows
processes to use both message-passing and shared-memory.
Communication is assumed to be lossless and provides FIFO semantics. 
The system has $n$ processes $\Pi= \{p_1,\ldots, p_{n}\}$ that can attain the roles of \textit{proposer} or \textit{acceptor}.
In the system, we assume that there are $p$ proposers and $n$ acceptors, where $1 < p < |\Pi|$. Processes can fail by crashing. Up to $p-1$ proposers and $\left \lfloor{\frac{n-1}2}\right \rfloor$ acceptors may fail.
As long as a process is alive, its memory is remotely accessible.
When a process crashes, its memory also crashes. In this case, subsequent memory operations do not return a response.
The system is asynchronous in that it can experience arbitrary delays.

\subsection{Consensus}\label{sec:consensus}

In the consensus problem, processes \textit{propose} individual values and eventually irrevocably \textit{decide} on one of them. Formally, Consensus has the following properties:
\begin{description}
	\item [Uniform agreement] If processes $i$ and $j$ decide $val$ and $val'$, respectively, then $val = val'$.
	\item [Validity]  If some process decides $val$, then $val$ is the input of some process.
    \item [Integrity] No process decides twice.
    \item [Termination] Every correct process that proposes eventually decides.
\end{description}

It is well known that consensus is impossible in the asynchronous model~\cite{fischer1985flp}. To circumvent this impossibility, an additional synchrony assumption has to be made. Our consensus algorithm provides safety in the asynchronous model and requires partial synchrony for liveness. For pedagogical reasons and in order to facilitate understanding, we implement our consensus algorithm by merging together the following abstractions:
\begin{itemize}
    \item \textit{Abortable Consensus}~\cite{cachin2011rachid}, an abstraction weaker than \textit{Consensus} that is solvable in the asynchronous model,
    \item \textit{Eventually Perfect Leader Election}~\cite{chandra1996weakest}, which relies on the weakest failure detector required to solve \textit{Consensus}.
\end{itemize}



\subsection{Abortable Consensus}\label{ss:abortable-consensus}

\emph{Abortable consensus} has the following properties:
\begin{description}
	\item [Uniform agreement] If processes $i$ and $j$ decide $val$ and $val'$, respectively, then $val = val'$.
	\item [Validity]  If some process decides $val$, then $val$ is the input of some process.
    \item [Termination] Every correct process that proposes eventually decides or abort.
    \item [Decision] If a single process proposes infinitely many time, it eventually decides.
\end{description}

Algorithm \ref{alg:rpc-abortable-consensus} solves Abortable Consensus and is based on Paxos. Processes are divided into two groups: proposers or acceptors. Proposers \textit{propose} a value for decision and acceptors \textit{accept} some proposed values. Once a value has been accepted by a majority of acceptors, it is decided by its proposer.

\begin{lstlisting}[columns=fullflexible,breaklines=true,keywords={while,if,else,return,do,for}, aboveskip=0pt, belowskip=0pt, float=ht!, caption={Abortable Consensus},label={alg:rpc-abortable-consensus}] 
@\textbf{Proposers execute:} \label{code:rpc-paxos-proposer}@
upon <Init>:
    decided = False
    proposal = id
    proposed_value = @$\bot$@
    
propose(value):
    proposed_value = value
    if not decided:
        if prepare():
            accept()
    
prepare():
    proposal = proposal + @$|\Pi|$@
    broadcast <Prepare | proposal>
    wait for a majority of <Prepared | ack, ap, av>
    if any av returned, replace proposed_value with av with highest ap
    if any not ack:
        trigger <Abort>
        return False
    return True
    
accept():
    broadcast <Accept | proposal, proposed_value>
    wait for a majority of <Accepted | mp>
    if any mp @$>$@ proposal:
        trigger <Abort>
    else:
        decided = true
        trigger <Decide | proposed_value>

@\textbf{Acceptors execute:} \label{code:rpc-paxos-acceptor}@
upon <Init>:
    min_proposal = 0
    accepted_proposal = 0
    accepted_value = @$\bot$@

upon <Prepare | proposal>:
    if proposal @$>$@ min_proposal:
        min_proposal = n
    reply <Prepared | min_proposal == n, accepted_proposal, accepted_value>

upon <Accept | proposal, value>:
    if proposal @$\geq$@ min_proposal:
        accepted_proposal = min_proposal = n
        accepted_value = value
    reply <Accepted | min_proposal>
\end{lstlisting}

Intuitively, the algorithm is split in two phases: the \textit{Propose} phase and the \textit{Accept} phase. During these phases, messages from the proposer are identified by a unique \textit{proposal number}.
The Prepare phase serves two purposes. First, the proposer gets a promise from a majority of acceptors that another proposer with a lower proposal number will fail to decide. Second, the proposer updates its proposed value using the accepted values stored in the acceptors. This way, if a value has been decided, the proposer will adopt it. The prepare phase can also trigger \texttt{Abort} if any acceptor in the majority previously made a promise to a higher proposal number.

If the proposer manages to complete the Prepare phase without aborting, it proceeds to the Accept phase. In this phase, the proposer tries to store its proposal value in a majority of acceptors. If it succeeds, it decides on that value. Otherwise, it ran obstructed and triggers \textit{Abort}.

A proof of correctness for algorithm \ref{alg:rpc-abortable-consensus} is given in \cite{cachin2011rachid}.

\subsection{From Abortable Consensus to Consensus}\label{sec:consensus-in-mp}

One can solve \textit{Consensus} by combining \textit{Abortable Consensus} together with \textit{Eventually Perfect Leader Election} ($\Omega$). In Abortable Consensus a proposer is guaranteed to decide, rather than abort, if it executes unobstructed. The role of $\Omega$ is to ensure this condition is ensured. Informally, $\Omega$ guarantees that eventually all correct proposers will consider a single one of them to be the leader. As long as a proposer is considering itself as the leader it keeps on proposing their value using \textit{Abortable Consensus}. Eventually, $\Omega$ will mark a single correct proposer as the leader, which will try to propose unobstructed and decide. The leader can then broadcast the decision to the rest of the proposers. Algorithm \ref{alg:rpc-consensus} provides the implementation of this idea. A proof of its correctness is given in \cite{cachin2011rachid}.

\begin{lstlisting}[columns=fullflexible,breaklines=true,keywords={while,if,else,return,do,for}, aboveskip=0pt, belowskip=0pt, float=ht! ,caption={Consensus from Abortable Consensus},label={alg:rpc-consensus}] 
@\textbf{Proposers execute:}@
upon <Init>:
    proposed_value = @$\bot$@
    leader = proposed = decided = False

upon <Trust | @$p_i$@>:
    if @$p_i$@ == self then leader = True
    else leader = False

propose(value):
    proposed_value = value
    while True:
        if leader and not proposed:
            proposed = True
            trigger <AbortableConsensus, Propose | proposed_value>
    
upon <AbortableConsensus, Decide | value>:
    broadcast <Decided | value>
    
upon <AbortableConsensus, Abort>:
    proposed = False
    
upon <Decided | value>:
    if not decided:
        decided = true
        trigger <Decide | value>
\end{lstlisting}

Notice that \textit{Abortable Consensus} differs from Consensus only in its liveness property. The former is essentially an obstruction free consensus implementation in which no proposer may decide under contention. However, Abortable Consensus retains the safety properties of Consensus. Thus, for the rest of this work we will concentrate on Abortable Consensus and we will transform it into a CAS-based algorithm.


\section{One-sided Consensus} \label{sec:cas-rpc}

In this section, we explain how to transform the two-sided algorithm presented in section \ref{ss:abortable-consensus} into a one-sided one. To do so, we first establish the equivalence between the RPCs used in algorithm \ref{alg:rpc-abortable-consensus} and CAS. 
Then, we take advantage of this equivalence and replace RPCs with CAS in algorithm \ref{alg:rpc-abortable-consensus}.
The resulting CAS-based Abortable Consensus is given in algorithm \ref{alg:cas-abortable-consensus}.

\subsection{One-sided obstruction-free RPC} \label{ss:cas-rpc}
Observe that algorithm \ref{alg:rpc-abortable-consensus} uses message passing (i.e. RPC) in a very specific form. The acceptors keep track of only three variables: \texttt{min\_porposal}, \texttt{accepted\_proposal} and \texttt{accepted\_value}. In both the Accept and the Prepare phases, acceptors atomically update these values upon a very simple condition (i.e., a simple comparison) and return some of them. In this section, we propose and prove a simple obstruction-free transformation to turn such RPCs into purely one-sided conditional writes using CAS.

\begin{lstlisting}[columns=fullflexible,breaklines=true,keywords={while,if,else,return,do,for}, aboveskip=0pt, belowskip=0pt, float=ht!, caption={CAS-based RPC},label={alg:cas-based-rpc}] 
rpc(x): @\label{code:rpc-call}@
    if compare(x, state): @\label{code:rpc-compare}@
        state = f(state, x)
    return projection(state)

cas-rpc(x):
    expected = fetch_state() @\label{code:cas-fetch}@
    if not compare(x, expected): @\label{code:cas-compare}@
        return projection(expected) @\label{code:cas-ret-1}@
    
    move_to = f(expected, x)
    old = cas(state, expected, move_to) @\label{code:cas-cas}@
    if old == expected: @\label{code:cas-compare-cas}@
        return projection(move_to) @\label{code:cas-ret-2}@
        
    abort()
\end{lstlisting}

In algorithm \ref{alg:cas-based-rpc}, we assume that the whole state of a process (i.e., all its variables) is stored in \texttt{state}. In the case of RPC (line \ref{code:rpc-call}), the caller sends \texttt{x} to the callee. The callee deterministically compares \texttt{x} with its state using \texttt{compare}. If the comparison succeeds, its state is deterministically mutated using the function \texttt{f}. In any case, the callee extracts part of its state using \texttt{projection} and returns it to the caller. By convention, \texttt{rpc} runs atomically.

We prove below that if the callee's \texttt{state} is accessible by the caller via shared memory, and \texttt{compare}, \texttt{f}, \texttt{projection} are known to the caller, then \texttt{rpc} and \texttt{cas-rpc} are strictly equivalent except in the case where \texttt{cas-rpc} aborts.

\begin{lemma} \label{lemma:equivalence}
If \texttt{cas-rpc} does not abort, \texttt{rpc} and \texttt{cas-rpc} are equivalent.
\end{lemma}

\begin{proof}
An execution of \texttt{rpc} solely depends on the value of \texttt{state} and the input value \texttt{x}. We denote such execution of \texttt{rpc} with $\langle state, x \rangle_{rpc}$. If an execution of \texttt{cas-rpc} does not abort, it solely depends on the value of \texttt{expected} fetched at line \ref{code:cas-fetch} and the input value \texttt{x}. We denote such execution of \textit{cas-rpc} with $\langle expected, x \rangle_{cas-rpc}$.

We show that any execution $\langle s, x \rangle_{rpc}$ is equivalent to the execution $\langle s, x \rangle_{cas-rpc}$ in the sense that both $rpc$ and $cas-rpc$ will have the same value of \texttt{state} and return the same projection at the end of their execution.

If an execution $\langle s_1, x \rangle_{rpc}$ makes the comparison at line \ref{code:rpc-compare} fail, then \texttt{state} is not modified and \texttt{projection($s_1$)} is returned.
In the execution $\langle s_1, x \rangle_{cas-rpc}$, the comparison at line \ref{code:cas-compare} will also fail (as the comparison is deterministic) and \texttt{projection($s_1$)} is also returned without modifying the remote state.
In this case, both executions are equivalent.

If an execution $\langle s_2, x \rangle_{rpc}$ makes the comparison at line \ref{code:rpc-compare} succeed, then \texttt{state} is modified to \texttt{f($s_2$, x)} and \texttt{projection(f($s_2$, x))} is returned.
In the execution $\langle s_2, x \rangle_{cas-rpc}$, the comparison at line \ref{code:cas-compare} will also succeed (since the comparison is deterministic). As the run is assumed to not abort, the CAS will succeed. Thus the remote state will atomically be updated from $s_2$ to \texttt{f($s_2$, x)} and \texttt{f($s_2$, x)} is also returned.
So, in this case, both executions are also equivalent.
\end{proof}

In addition, this transformation is safe in case of obstruction

\begin{lemma} \label{lemma:nose}
If \texttt{cas-rpc} aborts, it has no side effect.
\end{lemma}

\begin{proof}
If \texttt{cas-rpc} aborts, the comparison at line \ref{code:cas-compare-cas} failed. This in turn implies that the CAS at line failed and thus that \texttt{state} is unaffected by the execution.
\end{proof}

From lemmas \ref{lemma:equivalence} and \ref{lemma:nose}, \texttt{cas-rpc} exhibits all-or-nothing atomicity.
We now prove that such a transformation in obstruction-free.

\begin{lemma} \label{lemma:obstruction-freedom}
If \texttt{cas-rpc} runs alone (i.e., unobstructed), it does not abort.
\end{lemma}

\begin{proof}
Let's assume by contradiction that \texttt{cas-rpc} runs alone and aborts. For \texttt{cas-rpc} to abort, the comparison at line \ref{code:cas-compare-cas} must fail. This in turn implies that the CAS at line \ref{code:cas-cas} failed, i.e., that the current value of \texttt{state} does not match the \texttt{expected} one. For this to happen the \texttt{state} must have been updated between line \ref{code:cas-fetch} and \ref{code:cas-cas} by another process. This means that there was a concurrent execution, a contradiction.
\end{proof}









\subsection{A purely one-sided consensus algorithm} \label{sec:simple-one-sided-consensus}

\begin{lstlisting}[columns=fullflexible,breaklines=true,keywords={while,if,else,return,do,for}, aboveskip=0pt, belowskip=0pt, float=ht!,caption={CAS-based Abortable Consensus},label={alg:cas-abortable-consensus}]
@\textbf{Proposers execute:} \label{code:cas-paxos-proposer}@
upon <Init>:
    decided = False
    proposal = id
    proposed_value = @$\bot$@
    
propose(value):
    proposed_value = value
    if not decided:
        if prepare():
            accept()
    
prepare():
    proposal = proposal + @$|\Pi|$@
    execute in parallel cas_prepare(p, proposal) for p in Acceptors
    wait for a majority to abort or return <ack, ap, av>
    if any returned, replace proposed_value with av with highest ap
    if any aborted or not ack:
        trigger <Abort>
        return False
    return True
    
accept():
    execute in parallel cas_accept(p, proposal, proposal_value) for p in Acceptors
    wait for a majority to abort or return mp
    if any aborted or returned mp @$>$@ proposal:
        trigger <Abort>
    else:
        decided = true
        trigger <Decide | proposed_value>

cas_prepare(p, proposal):
    expected = fetch_state(p)
    if not proposal @$>$@ expected.min_proposal:
        return <false, expected.accepted_proposal, expected.accepted_value>
    
    move_to = expected
    move_to.min_proposal = proposal
    old = cas(p.state, expected, move_to)
    if old == expected:
        return <true, expected.accepted_proposal, expected.accepted_value>
        
    abort()
    
cas_accept(p, proposal, value):
    expected = fetch_state(p)
    if not proposal @$\geq$@ expected.min_proposal:
        return expected.min_proposal
    
    move_to = expected
    move_to.min_proposal = proposal
    move_to.accepted_proposal = proposal
    move_to.accepted_value = value
    old = cas(p.state, expected, move_to)
    if old == expected:
        return expected.min_proposal
        
    abort()
    
@\textbf{Acceptors execute:} \label{code:cas-paxos-acceptor}@
upon <Init>:
    state = @$\{ min\_proposal: 0, accepted\_proposal: 0, accepted\_value: \bot\}$@
\end{lstlisting}

Algorithm \ref{alg:cas-abortable-consensus} implements Abortable Consensus by replacing RPCs in algorithm \ref{alg:rpc-abortable-consensus} with the one-sided obstruction-free RPCs introduced in \ref{ss:cas-rpc}. As this new algorithm relies on CAS for comparisons, it aborts in situations where the original algorithm would have succeeded. For example, consider the following execution: Let proposers $P_1$ and $P_2$ concurrently initiate the Prepare phase with respective proposals $1$ and $2$. Both fetch the remote state and get $\langle 0, 0, \bot \rangle$. Then, $P_1$ succeeds at writing its proposal to acceptor $A_1$. Later on, the CAS of $P_2$ fails at $A_1$ as the value is now $\langle 1, 0, \bot \rangle$ instead of the expected $\langle 0, 0, \bot \rangle$. Thus, $P_2$ aborts even if it had a larger proposal number than $P_1$. The more relaxed comparison in the original algorithm would not have caused $P_2$ to abort.



\begin{lemma}\label{lemma:ccdp}
Algorithm \ref{alg:cas-abortable-consensus} preserves the decision property of Abortable Consensus.
\end{lemma}

\begin{proof}
If a single process proposes infinitely many time, it will eventually run one-sided RPCs obstruction-free. By lemma \ref{lemma:obstruction-freedom}, this guarantees that eventually the one-sided RPCs will terminate without aborting. In such case, lemma \ref{lemma:equivalence} guarantees the execution to be equivalent to one of the original algorithm. Thus, the transformation preserves the decision property of the original algorithm.
\end{proof}

\begin{lemma}\label{lemma:cctp}
Algorithm \ref{alg:cas-abortable-consensus} preserves the termination property of Abortable Consensus.
\end{lemma}

\begin{proof}
Assuming a majority of correct acceptors, all CASes will eventually return or abort. Due to the absence of loops or blocking operations (apart from waiting for a reply from a majority of acceptors), a proposer that invokes \texttt{propose} will either abort or decide.
\end{proof}

The only execution difference between both algorithms is that some executions of the transformed algorithm may abort, where the original one would not. Nevertheless, aborting does not violate safety.

\begin{lemma}\label{lemma:ccsp}
Algorithm \ref{alg:cas-abortable-consensus} preserves the safety properties of Abortable Consensus.
\end{lemma}

\begin{proof}
Assume by contradiction that adding superfluous abortions in the original algorithm can violate safety. In a first execution $E_1$, processes \{$P_1$, ..., $P_n$\} deviate from the original algorithm and abort at times \{$t_1$, ..., $t_n$\} after which the global state is \{$S_1$, ..., $S_n$\}. At some point, safety is violated. In a second execution $E_2$, processes \{$P_1$, ..., $P_n$\} crash at times \{$t_1$, ..., $t_n$\} after which the global state is \{$S_1$, ..., $S_n$\}. As the original algorithm tolerates arbitrarily many proposer crash failures, safety is not violated. Proposers cannot distinguish both executions. Thus, safety cannot be violated, hence a contradiction. Thus, adding superfluous aborts preserves safety and algorithm \ref{alg:cas-abortable-consensus} preserves safety.  
\end{proof}

\begin{theorem}
By lemmas \ref{lemma:cctp}, \ref{lemma:cctp}, \ref{lemma:ccsp}, algorithm \ref{alg:cas-abortable-consensus} implements Abortable Consensus. \qed
\end{theorem}

\subsection{Streamlined one-sided algorithm}

Section \ref{sec:simple-one-sided-consensus} introduced a simple one-sided consensus algorithm built by replacing RPCs with weaker (i.e., abortable) one-sided RPCs and proved its correctness. The resulting algorithm exhibits flagrant inefficiencies that can can be fixed to reduce the number of network operations to 2 CASes in the common case.

First, it is not required to fetch the remote state at the start of each RPC. As it is safe to have stale \texttt{expected} states, it is safe to use optimistic predicted states deduced from previous CASes. Predicted states can thus be initialized to $\langle 0, 0, \bot \rangle$ and updated each time a CAS completes (either succeeding or not). Moreover, wrongly predicting  states can only result in superfluous aborts which have been proven to be safe by lemma \ref{lemma:ccsp}. Thus, it is safe to optimistically assume that onflight CAS will succeed.

Second, in the Prepare phase, \texttt{proposal} can be increased upfront to be higher than any predicted remote \texttt{min\_proposal} to reduce predictable abortions.

Algorithm \ref{alg:streamlined-cas-abortable-consensus} gives the algorithm obtained after applying these optimisations.

\begin{lstlisting}[columns=fullflexible,breaklines=true,keywords={while,if,else,return,do,for}, aboveskip=0pt, belowskip=0pt, float=ht! ,caption={Streamlined One-sided Abortable Consensus},label={alg:streamlined-cas-abortable-consensus}]
@\textbf{Proposers execute:} \label{code:sld-cas-paxos-proposer}@
upon <Init>:
    predicted = @$\{ min\_proposal: 0, accepted\_proposal: 0, accepted\_value: \bot\}^{|Acceptors|}$@
    decided = False
    proposal = id
    proposed_value = @$\bot$@
    
propose(value):
    proposed_value = value
    if not decided:
        if prepare():
            accept()
    
prepare():
    for p in Acceptors:
        while predicted[p].min_proposal @$\geq$@ proposal:
            proposal = proposal + @$|\Pi$|@
    
    reads = @$\bot^{|Acceptors|}$@
    move_to = @$\bot^{|Acceptors|}$@
    
    for p in Acceptors:
        move_to[p] = (proposal, predicted[p].accepted_proposal, predicted[p].accepted_value)
        async reads[p] = cas(slot@$_p$@, predicted[p], move_to[p])
    wait until majority of slots are read
    
    for p in Acceptors:
        if reads[p] @$\in$@ {predicted[p], @$\bot$@}:
            predicted[p] = move_to[p] @\label{code:optimistic-update:1}@
        else:
            predicted[p] = reads[p]
            
    if any CAS failed (predicted[p] @$\neq$@ read[p]):
        trigger <Abort> @\label{code:sld-abort:1}@
        return false

    proposed_value = predicted[.].accepted_value with highest accepted_proposal or proposed_value if none
    return true
    
accept():
    move_to = (proposal, proposal, proposed_value)
    for p in Acceptors:
        async reads[p] = cas(slot@$_p$@, predicted[p], move_to)
    wait until majority of slots are read
    
    if any CAS failed:
        for p in Acceptors:
            if reads[p] @$\in$@ {predicted[p], @$\bot$@}:
                predicted[p] = move_to @\label{code:optimistic-update:2}@
            else:
                predicted[p] = reads[p]
        trigger <Abort> @\label{code:sld-abort:2}@
        return

    decided = true
    trigger <Decide | proposed_value>
    
@\textbf{Acceptors execute:} \label{code:sld-cas-paxos-acceptor}@
upon <Init>:
    state = @$\{ min\_proposal: 0, accepted\_proposal: 0, accepted\_value: \bot\}$@
\end{lstlisting}

Said optimisations also preserve liveness. Let's assume that a single proposer runs infinitely many time. Eventually, it will run obstruction-free. In the worst case, it will each time abort at line \ref{code:sld-abort:1} or \ref{code:sld-abort:2} because of a single wrongly guessed remote state and update its prediction. The optimistic update of expected states at lines \ref{code:optimistic-update:1} and \ref{code:optimistic-update:2} and the FIFO property of links provide that, once a remote state is correctly guessed, any later CAS will succeed. Thus, after at most $n$ runs all CASes will succeed and the proposer will decide.
\section{Practical Considerations} \label{sec:practical-considerations}

So far we have showed a new algorithm for doing RDMA-based consensus using CAS. Our algorithm presents a single instance of consensus, however most practical systems require to run consensus over and over.

State Machine Replication (SMR) replicates a service (e.g., a key-value storage system)
across multiple physical servers called \emph{replicas}, such that the system remains available and consistent even if some servers fail.
SMR provides strong consistency in the form of linearizability~\cite{HW90}.
A common way to implement SMR, which we adopt in this paper, is as follows: each replica
  has a copy of the service software and a log. 
The log stores client requests.
We consider non-durable SMR systems~\cite{ramcloud,li2016just,curp,jin2018netchain,istvan2016consensus,ipipe},
  which keep state in memory only, without logging updates to stable storage. 
Such systems make an item of data reliable by keeping copies of it in the memory of several nodes. Thus, the data remains recoverable as long as there are fewer simultaneous node failures than data copies~\cite{poke2015dare}.

A consensus protocol ensures that all replicas agree on what request is stored in each 
  slot of the log.
Replicas then apply the requests in the log (i.e., execute the 
corresponding operations), in log order.
Assuming that the service is deterministic, this ensures all replicas remain
  in sync.
We adopt a leader-based approach, in which a dynamically elected replica 
  called the \emph{leader} communicates with the clients and sends back responses
  after requests reach a majority of replicas.
As already stated, we assume a \textit{crash-failure} model: servers may fail by crashing, after which they
  stop executing.
  
Velos is the implementaion of Algorithm \ref{alg:streamlined-cas-abortable-consensus} in the form of SMR. Implementing Velos leads to practical considerations and additional hardware-specific optimizations that are not present in Algorithm \ref{alg:streamlined-cas-abortable-consensus}.

\subsection{Pre-preparation}

Unfortunately, our CAS approach prevent us from using the multi-Paxos optimisation\cite{pml2007, lamport2001paxos, mazieres2007paxos}. Thus, each consensus slot must be prepared individually. Nevertheless, a leader can prepare slots in advance so that the it can decide running only the accept phase on its critical path. In this case, the leader decides in a single CAS RTT. Doing so requires a stable leader. Switching to another leader requires re-preparing slots. Fortunately, this will most likely succeed in a single CAS RTT by optimistically predicting remote slots to have been prepared by the failed leader.

\subsection{Limited CAS size}

Algorithm \ref{alg:streamlined-cas-abortable-consensus} assumes that a state can fit within a single CAS. Current RDMA NICs support CAS up to 8 bytes. As both \texttt{min\_proposal} and \texttt{accepted\_proposal} must be the same size, both fields are limited to at most 31 bits with 2 bits being dedicated to storing the \texttt{accepted\_value}.

One may be afraid of proposal fields overflowing (either breaking safety or decision). In such an extremely unlikely case (and actually more for completeness), the abstraction can fallback to traditional RPC. Switching to RPC can be safely implemented as follows: Once the RDMA-exposed  \texttt{min\_proposal} of an acceptor reaches $2^{31} - |\Pi|$, every proposer switches to RPC to communicate with this specific acceptor. At a regular time interval, acceptors checks their \texttt{state} and, if above the threshold, execute algorithm \ref{alg:rpc-abortable-consensus} with \texttt{min\_proposal}, \texttt{accepted\_proposal} and \texttt{accepted\_value} variables initiated to match \texttt{state}.

Another concern is the limited \texttt{accepted\_value} field size. Depending on the upper application, the value can be inlined within the CAS. Otherwise, a simple indirection can be put in place. Instead of deciding on the proposed value itself, one can decide on the proposer's id. This can be achieved by RDMA-writing the actual value to a majority of acceptors in a dedicated write-exclusive slot before running the accept phase. RDMA RC QP FIFO semantics can be leverages to do so at minimimum cost. The write WQE can be prepended to the Accept WQE, made unsignaled, and posted with Doorbell batching. If the CAS completes, then the value was written at a majority of acceptors and will always be recoverable. 

\subsection{Device Memory}

Modern RNICs allows RDMA exposure of their own internal memory. This feature is called \textit{Device Memory} (DM). In Mellanox' Connect-X6, the available memory is slightly larger than 100KiB. RDMAing this memory is faster than accessing the main memory as it removes an extra PCIe communication from the critical path. All RDMA verbs benefit from a speedup. Moreover, DM reduces atomic verbs contention. As one-sided Paxos makes acceptors fundamentally passive, DM can be used without additional cost. DM can also be leveraged when RDMA-writing the actual value as described in the previous section to PCIe-transfer the payload only once.

\subsection{Piggybacking decisions}

Consensus slots can be augmented with an additional \texttt{previous\_decision} value. This way, if every node endorses both the proposer and the acceptor roles, it can learn that a value has been decided for slot $s-1$ simply by reading its local slot $s$.
With this strategy, values are decided in a single CAS and decisions learned with no additional communication in the common case. 

\subsection{Unavailable features}

Modern RNICs lack some features and performances that would make one-sided Paxos even more appealing. We believe that these features can reasonably be expected to be provided by future RNICs.

First are global CASes. CX6 atomics are only guaranteed to be atomic with other operations executed by the same HCA. This prevent us from using CPU-side CASes to update the \texttt{state} of the proposer, which could save one MMIO ($\approx$300ns). 

Second are unalligned CASes. Currently, CASes cannot overlap 2 consecutive slots. Such a feature would allow a proposer to run the Accept phase for slot $s$ and the Prepare phase for slot $s+1$ in a single CAS, definitely removing the need for batch pre-preparation from the critical path.

Third, even using Device Memory, CAS exhibits a latency twice higher than RDMA WRITEs, which makes Mu twice as fast in the failure-free scenario.
\section{Implementation}\label{sec:impl} \label{sec:opt}

Our algorithm is implemented in 1046 lines of C++17 code (CLOC~\cite{tool-cloc}).
It uses the \emph{ibverbs} library for RDMA over Infiniband and it relies on the reusable abstractions provided by the source code of Mu. Furthermore, we implemented the optimizations mentioned in section~\ref{sec:practical-considerations}, apart from ``Piggybacking decisions''.

We implement a new leader election algorithm that departs from Mu's design. Our algorithm relies on modifications of the Linux kernel and is composed of two separate modules, namely the \textit{interceptor} and the \textit{broadcaster} module. In Linux, when a process crashes control is transferred to the kernel which takes care of cleaning up the process. Our kernel modifications allow processes to instruct the kernel---via a \texttt{prctl} system call---that upon failure the interceptor module should be invoked during the cleaning up. The interceptor module is then responsible for invoking the broadcaster module. The broadcasting module broadcasts a message saying that the process has crashed. This message is picked up by correct processes that subsequently stop trying to communicate with the crashed process. Our kernel modifications and kernel modules implementation span in 352 lines of C code.



\section{Evaluation} \label{sec:evaluation}
Our goal is to evaluate whether our implementation can achieve consensus within a few microseconds and whether it handles failures with the least amount of delay. Concretely, with our evaluation we aim at answering the following questions:
\begin{itemize}
    \item Does our implementation imposes a small overhead compared to Mu in the common case?
    \item What is the total failover time during a leader crash?
\end{itemize}

We evaluate our system on a 3-node cluster, the details of which are given in Table~\ref{tab:hwspecs}.

In the reported numbers we show 3-way replication, which accounts
for most real deployments~\cite{hunt2010zookeeper}. 
In all of our experiments we ignore the existence of a client issuing requests to our system.

{\footnotesize
\begin{table}[ht!]
    \centering
    \caption{Hardware details of machines.}
	\begin{tabular}{cm{0.64\linewidth}}
	\toprule
\textbf{CPU}		&   2x Intel(R) Xeon(R) Gold 6244 CPU @ 3.60GHz \\
\textbf{Memory}	&   2x 128GiB \\
\textbf{NIC}		&   Mellanox ConnectX-6 \\
\textbf{Links}   &   100 Gbps Infiniband \\
\textbf{Switch}  &   Mellanox MSB7700 EDR 100 Gbps  \\
\textbf{OS}      &   Ubuntu 20.04.2 LTS \\
\textbf{Kernel}  &   \texttt{5.4.0-74-custom} \\
\textbf{RDMA Driver} & Mellanox OFED \texttt{5.3-1.0.0.1} \\
	    \bottomrule
	\end{tabular}
    \label{tab:hwspecs}
\end{table}
}

We measure time using the POSIX \texttt{clock\_gettime} function, with the \texttt{CLOCK\_MONOTONIC} parameter. 
In our deployment, the resolution and overhead of \texttt{clock\_gettime} is around $16$--$20ns$~\cite{uarch-bench}.

\subsection{Common-case Replication Latency}
We first evaluate the latency of our system under no leader failure. We measure the latency at the leader. Latency refers to the time it takes for \texttt{propose} of Algorithm~\ref{alg:streamlined-cas-abortable-consensus} to execute. We show the time it takes to replicate messages of different sizes in Figure~\ref{fig:latency}.

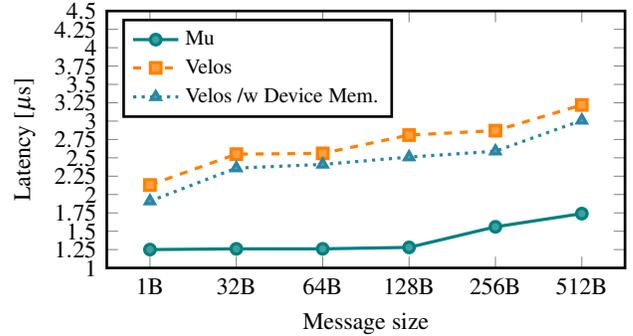
\begin{figure}
\centering
\begin{tikzpicture}
	
    \begin{axis}[
        height=5cm,
        width=\linewidth,
        ytick scale label code/.code={},
        xlabel={Message size},
        ylabel={Latency [$\mu$s]},
        ymin=1,
        ymax=4.5,
        ytick distance=0.25,
        xtick=data,
        xticklabels from table={data/propose_latency.txt}{size},
        xmode=normal,
        ymode=normal,
        legend pos=north west,
        legend style={
            cells={anchor=west}}
    ]
        
    \addplot table[x=x-pos,y=mu] {data/propose_latency.txt};
    \addlegendentry{Mu};
    
    \addplot table[x=x-pos,y=velos] {data/propose_latency.txt};
    \addlegendentry{Velos};

    \addplot table[x=x-pos,y=velosdm] {data/propose_latency.txt};
    \addlegendentry{Velos /w Device Mem.};

    \end{axis}
	
\end{tikzpicture}
\\

\caption{Median replication latency of Velos compared to Mu for different message sizes.}
\label{fig:latency}
\end{figure}

For messages of 1 byte, Velos replicates a request using only a RDMA CAS. On the other side, Mu always uses a single RDMA WRITE to replicate. For messages up to 128 bytes Mu manages to inline the message in the RDMA request, thus exhibiting almost constant replication latency. RDMA WRITE and RDMA CAS are both one-sided operations, but they have different latency. Sending 3 RDMA CASes and waiting for a majority of replies costs around 1.9$\mu s$, while the same communication pattern with RDMA WRITEs costs 1.25$\mu s$. This is seen from the time difference for 1B messages.

For larger payload sizes, Velos replicates a request using a RDMA CAS and an additional RDMA WRITE. In other words, Velos does exactly what Mu does, apart from having an extra RDMA CAS for every replicated message. Given that the latency of RDMA CAS is constant (in the absence of CAS-contention, i.e., when having a stable leader), the impact of the overhead in the replication latency of Velos compared to Mu diminishes for larger message sizes.

Figure~\ref{fig:latency} also demonstrates the effect of using Device Memory. When Velos relies exclusively on Device memory for its RDMA WRITEs and RDMA CASes, it gains 200$ns$ on replication latency.

\subsection{Fail-over Time}
During a leader crash, Velos replicas receive the failure detection event and a new leader is elected. The new leader immediately starts replicating new messages among itself and the remaining replica $R$.

Figure~\ref{fig:tput-leader-failure} evaluates the time it takes for $R$ to discover a new replicated value in its log. When a stable leader replicates requests, the throughput is at around 42 requests per 100$\mu s$. In other words, $R$ discovers a new entry in its log approximately every 2.5$\mu s$. When the leader fails the throughput drops to 0 and replica $R$ discovers a new value after approximately 65$\mu s$. The subsequent few replication requests from the new leader take between 3$\mu s$ to 3.6$\mu s$, which is evident in the throughput curve. When the new leader takes over, the throughput curve exhibits a not-so-steep trajectory, before stabilizing again at the throughput of 42 requests per 100$\mu s$ This is because the new leader's cache is cold and initial replication requests result in higher replication latency. Soon after the new leader manages to replicate new requests in approximately 2.5$\mu s$.

Comparing Velos to Mu, the latter faster than Velos in the common case but slower during leader change.
Mu manages to replicate requests using a single WRITE but it relies on permissions to handle concurrent leaders during leader failure. The cost of changing permissions, as presented in Mu, is approximately 250$\mu s$ just for changing the leader. Mu requires an additional 600$\mu s$ to detect the leader failure. On the other side, Velos requires approximately 30$\mu s$ to detect a leader failure and an additional 35$\mu s$ for the new leader to successfully replicate the next request. In other words, Velos is approximately 1.2 or 1.5 $\mu s$ slower than Mu in the common case, depending on whether it relies or not in Device Memory. Velos is significantly faster than Mu during leader change, as it detects a leader failure 20 times faster than Mu and replicates the next request 7 times faster than Mu. Overall, Velos is 13 times faster than Mu during leader change.
\begin{figure}
\centering
\begin{tikzpicture}
	
    \begin{axis}[
        height=5cm,
        width=\linewidth,
        ytick scale label code/.code={},
        xlabel={Time ($\mu s$)},
        ylabel={Throughput [requests per 100 $\mu$s]},
        xmode=normal,
        ymode=normal,
        legend pos=north west,
        legend style={
            cells={anchor=west}}
    ]
        
    \addplot [mark=none] table[x=idx, y=tput] {data/throughput_leader_failure.txt};

    \end{axis}
	
\end{tikzpicture}
\\

\caption{Throughput under leader failure (for message size is 1 byte) measured from the remaining replica. }
\label{fig:tput-leader-failure}
\end{figure}
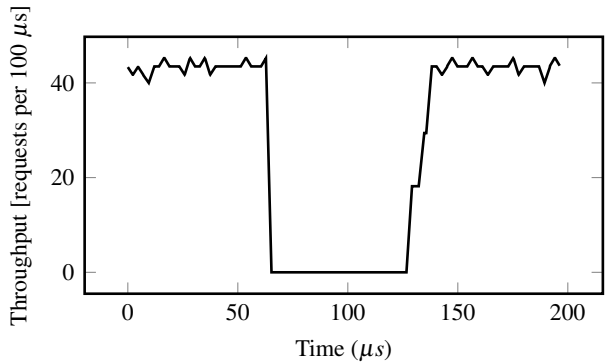


\section{Conclusion} \label{sec:conclusion}
Consensus is a classical distributed systems abstractions that is widely used in the data centers. RDMA enables consensus to achieve lower decision latency not only due to its intrinsic latency characteristics as a network fabric, but also due to its semantics. RDMA semantics such as permission changes improve the latency of consensus in the common case, as they enable consensus to decide by using a single RDMA WRITE. However, the non-common case during which failures still exhibited latency in the order of a millisecond.

Velos is a state machine replication system that can replicates requests in a few microseconds. It relies a modified Paxos algorithm that replaces RPCs (i.e. message passing) with RDMA Compare-and-Swap. As a result, Velos exhibits common-case latency that is competitive to Mu's latency and shines during failure. Velos manages to switch to a new leader and start replicating new requests in under 65 $\mu s$, meaning that Velos has an extra 9 of availability for the same number of expected failures, compared to Mu.


\bibliographystyle{plain}
\bibliography{references}

\end{document}